\documentclass[12pt]{article}
\newcommand{\version}{\today}

\usepackage{amsthm,amsfonts,amsmath, amscd,dsfont}
\usepackage{mathrsfs}

\swapnumbers
                              
\theoremstyle{plain}
\newtheorem{thm}{THEOREM}[section]
\newtheorem{lm}[thm]{LEMMA}

\theoremstyle{definition}

\newtheorem{exam}[thm]{EXAMPLE}
\newtheorem{remark}[thm]{Remark}
\theoremstyle{remark}
\newcommand{\upchi}{\raise1pt\hbox{$\chi$}}
\newcommand{\R}{{\mathord{\mathbb R}}}

\newcommand{\tr}{{\rm Tr}}
\numberwithin{equation}{section}
\newcommand{\un}{{\rm 1\kern -2.5pt l}}

\begin{document}
%%%%%%%%DRAFT%%%%%%%
\iffalse
\fi
%%%%%%%%%%%%%%%%%%%%%%

\markboth{\scriptsize{CL \version}}{\scriptsize{CL \version}}

\def\mn{{\bf M}_n}
\def\hn{{\bf H}_n}
\def\hnp{{\bf H}_n^+}
\def\hmnp{{\bf H}_{mn}^+}
\def\h{{\cal H}}

\title{A trace  inequality of Ando, Hiai and Okubo and a monotonicity property of the Golden-Thompson inequality}
\author{\vspace{5pt} Eric A. Carlen$^1$ and
Elliott H. Lieb$^{2}$ \\
\vspace{5pt}\small{$1.$ Department of Mathematics, Hill Center,}\\[-6pt]
\small{Rutgers University,
110 Frelinghuysen Road
Piscataway NJ 08854-8019 USA}\\
\vspace{5pt}\small{$2.$ Departments of Mathematics and Physics, Jadwin
Hall,} \\[-6pt]
\small{Princeton University, Princeton, NJ
  08544}\\
 }
\date{\version}
\maketitle
\footnotetext                                                                         
[1]{\, Work partially
supported by U.S. National Science Foundation
grant DMS 2055282..   }
\footnotetext
[2]{ \, \copyright \,2022  by the authors. This paper may be reproduced, in its
entirety, for non-commercial purposes.}

\begin{abstract}

The Golden-Thompson trace inequality which states that $Tr\, e^{H+K} \leq Tr\, e^H e^K$ has proved to be very useful in
quantum statistical mechanics. Golden used it to show that the classical free energy is less than the quantum one.
Here  we make this G-T inequality more explicit by proving that for some operators, notably the operators of interest in quantum mechanics,
$H=\Delta$ or $H= -\sqrt{-\Delta +m}$ and $K=$ potential,
$Tr\, e^{H+(1-u)K}e^{uK}$ is  a monotone increasing function of the parameter $u$ for $0\leq u \leq 1$.
Our proof utilizes an inequality of Ando, Hiai and Okubo (AHO): $Tr\, X^sY^tX^{1-s}Y^{1-t} \leq
Tr\, XY$ for positive operators X,Y and for $\tfrac{1}{2} \leq s,\,t \leq 1 $ and $s+t \leq \tfrac{3}{2}$. The obvious conjecture that this inequality should hold up to $s+t\leq 1$, was proved false by Plevnik. We give a different proof of AHO and also give more counterexamples
in the $\tfrac{3}{2}, 1$ range. More importantly we show that the inequality  conjectured in AHO does indeed hold in this
range if $X,Y$ have a certain positivity property -- one which does hold for quantum mechanical operators, thus enabling us to prove our G-T monotonicity theorem.

\end{abstract}

%%%\medskip
\leftline{\footnotesize{\qquad Mathematics subject classification numbers:  47A63, 15A90}}
\leftline{\footnotesize{\qquad Key Words: convexity, concavity, trace inequality, 
entropy, operator norms}}

\section{Introduction} \label{intro}

In 2000, Ando, Hiai and Okubo \cite{AHO} (AHO) considered several inequalities for traces of products of two
positive semidefinite matrices $X$ and $Y$, of which the two simplest were
\begin{equation} \label{1.1}
     |\tr [ X^sY^t X^{1-s}Y^{1-t}] |   \leq \tr[ XY]
\end{equation}
and
\begin{equation} \label{1.1b}
\tr  [X^{1/2}Y^{1/2}X^{1/2}Y^{1/2} ] \leq     |\tr  [X^sY^t X^{1-s}Y^{1-t}] |   
\end{equation}
with $1/2\leq s \leq 1$ and $1/2\leq t\leq 1$.

Note that the absolute value, or at least a real part is necessary for either \eqref{1.1} or \eqref{1.1b} to make sense; $\tr  [X^sY^t X^{1-s}Y^{1-t}]$ may be a complex number. 

Ando Hiai and Okubo succeeded in proving both inequalities when $X$ and $Y$ were $2$ by $2$ matrices, or more generally, 
when both $X$ and $Y$ have at most two distinct eigenvalues \cite[Corollary 4.3]{AHO}.  They also proved \eqref{1.1}
 when $ s+t \leq 3/2$, but could only prove \eqref{1.1b} when either $s=1/2$ or $t=1/2$. 
 They  raised the question as to whether the inequalities \eqref{1.1}  and \eqref{1.1b} 
hold over the entire range $1\leq s+t\leq 2$. In addition to proving the positive results mentioned above (and some generalizations discussed below) they remarked that the behavior of the function $(s,t) \mapsto  |\tr  [X^sY^t X^{1-s}Y^{1-t} ]| $ on the whole interval $[1/2,1]\times[1/2,1]$ is ``is rather complicated for general $n\times n$ positive semidefinite matrices''. 

The question they raised attracted the attention of other researchers.  In particular, Bottazzi, Elencwajg, Larotonda and Varela \cite{BELV} gave another proof, for the case $s=t$, that \eqref{1.1} is valid for $s+t \leq 3/2$.   Instead of the majorization techniques used in \cite{AHO}, they used the Lieb-Thirring inequality and the H\"older inequality for matrix trace norms. Using these tools, they showed that for $z = 1/4 + iy$ or $z= 3/4+iy$, $y\in \R$, 
\begin{equation}\label{BELV1}
 |\tr [ X^zY^z X^{1-z}Y^{1-z}] |   \leq \tr [XY]
\end{equation}
and then used then used the maximum modulus principle to conclude that \eqref{1.1} is valid for $s=t$, $1/4 \leq t \leq 3/4$. Moreover, they proved that unless $A$ and $B$ commute, this inequality is strict, and thus for any given $X$ and $Y$, the inequality extends to a wider interval, depending on $X$ and $Y$.   However,  16 years after the original work of Ando, Hiai and Okubo, Plevnik \cite{P}  finally found a counterexample to the conjectured inequality  \eqref{1.1} in
the missing range $3/2 \leq s+t\leq 2$, as well as a counterexample to  \eqref{1.1b}.

We, unaware of these developments, attempted  to show a monotonicity property 
for the  Golden-Thompson  inequality \cite{G65,T65}
and were led to exactly the same inequality that \cite{AHO} had discussed 22 years earlier. Our proof for the 
$1\leq s+t\leq 3/2$ range is a little different and we shall give that proof here.  We also identify interesting conditions on $X$ and $Y$ under which \eqref{1.1} and \eqref{1.1b} do hold for all $0 \leq s,t \leq 1$, and apply this to prove our conjecture on the Golden-Thompson inequality in these cases.
We shall also give a systematic construction of  counterexamples for the $3/2 \leq s+t\leq 2$ range that
complement the example in \cite{P} and show that not only is \eqref{1.1b} false, it is even possible for $\tr  [X^sY^t X^{1-s}Y^{1-t}]$ to be negative when $X$ and $Y$ are real positive semidefinite matrices.

\section{Conditions for validity of the AHO inequalities}

We recall the Lieb-Thirring inequality \cite{LT76}, which says that for all $r \geq 1$,  and any positive semidefinite $n\times n$ matrices,
\begin{equation}\label{LT1}
\tr[(B^{-1/2}AB^{1/2})^r] \leq \tr[A^rB^r]\ .
\end{equation}
Later, Araki \cite{A90} proved that the inequality reverses for $0 < r < 1$.  It was shown by Friedlander and So \cite{FS}  that for $r>1$, the inequality is strict unless $A$ and $B$ commute. 

In the following, and in the whole of this paper, $X$ and $Y$ are positive semidefinite matrices. We will use \eqref{LT1} to estimate $\|X^{1/p}Y^{1/p}\|_p$ for various values of $p\geq1$. Since
$$
\|X^{1/p}Y^{1/p}\|_p^p = \tr[ (Y^{1/p}X^{2/p}Y^{1/p})^{p/2}]\ ,
$$
we may apply \eqref{LT1} to get an upper bound on  $\|X^{1/p}Y^{1/p}\|_p$ taking $r =p/2$ provided $p/2 \geq 1$, or equivalently $1/p \leq 1/2$. (Otherwise, by Araki's complement to \eqref{LT1}, we would get a lower bound.)  In summary:
\begin{equation}\label{LT2}
\|X^{1/p}Y^{1/p}\|_p^p \leq \tr[XY]  \quad{\rm for\ all}\quad 0 < 1/p \leq1/2\ .
\end{equation}

As in \cite{BELV}, we shall use the generalized H\"older inequality for trace norms (see, e.g., Simon's book \cite{S05}.  
For any $3$ $n\times n$ matrices $A$, $B$ and $C$, and  any  $p,q,r \geq 1$ with $1/p + 1/q + 1/r =1$,
\begin{equation}\label{holder}
|\tr[ABC]| \leq \|ABC\|_1 \leq \|A\|_p\|B\|_q\|C\|_r\ .
\end{equation}
(This generalizes in the obvious way to products of arbitrarily many matrices.)

The next theorem is a small generalization of the result in \cite{AHO} in that we consider $4$ positive semidefinite  matrices instead of only $2$.

\begin{thm}\label{AHO1}    Let $X$, $Y$, $Z$, and $W$ be positive semidefinite, and let $1/2 \leq s,t$, $t+ s \leq 3/2$. Then 
\begin{equation}\label{q2b2}
\left|\tr[X^tY^s Z^{1-t}W^{1-s}]\right|  \leq  (\tr[XY])^{t+s-1}(\tr[YZ])^{1-t}  (\tr([WX]))^{1-s} 
\end{equation}
In particular, taking $Z=X$ and $W= Y$, we obtain \eqref{1.1} under these conditions on $s$ and $t$. 
\end{thm}

\begin{proof}  Since $s,t \geq 1/2$,  $t \geq 1-s$.  Write $t = (1- s) + (t+s-1)$, and both summands are non-negative. By cyclicity of the trace,
\begin{eqnarray*}
 \tr[X^tY^s Z^{1-t}W^{1-s}] &=& \tr[X^{t+s-1}Y^s Z^{1-t}W^{1-s}X^{1-s}]\\
  &=& \tr[(X^{t+s-1} Y^{t+s-1}) (Y^{1-t}Z^{1-t}) (W^{1-s}X^{1-s})]\ .
  \end{eqnarray*}
  Define $r_1 := t+s-1$, $r_2 :=1-t$ and $r_3 := 1-s$.  Then we have
  $$ \tr[X^tY^s Z^{1-t}W^{1-s}]   = \tr[ (X^{r_1}Y^{r_1}) (Y^{r_2}Z^{r_2}) (W^{r_3}X^{r_3})]\ .$$
By what was noted above, $r_1,r_2,r_3 \geq 0$, and of course $r_1+r_2+r_3 =1$. 
Thus by H\"older's inequality
\begin{equation}\label{holdlim}
\left|  \tr[X^tY^s Z^{1-t}W^{1-s}]    \right|  \leq  \|X^{r_1}Y^{r_1}\|_{1/r_1}  \|Y^{r_2}Z^{r_2}\|_{1/r_2}  \|W^{r_3}X^{r_3}\|_{1/r_3}\  .
\end{equation}
We may now apply \eqref{LT2} provided $r_1$, $r_2$ and $r_3$ are all no greater than $1/2$.
Since  $s,t \geq 1/2$, it is always the case that  $r_2,r_3 \leq 1/2$, while $r_1\leq 1/2$ if and only if $t + s \leq \frac32$.  Hence under this condition  \eqref{q2b2} is proved. 
\end{proof}

\begin{remark}The assumption that the two powers of $X$ sum to $1$ is not a real restriction. Given
two arbitrary positive powers $a,b$ we may rename $X^{a+b}$  to be $X$, and define $s := \max\{a,b\}/(a+b)$, and similarly for $Y$. 
\end{remark}

\begin{remark} In \cite{AHO},  Theorem~\ref{AHO1} was generalized to $n$ $X$'s and $n$ $Y$'s as follows, and our method 
of proof of Theorem \ref{AHO1} using the Lieb-Thirring inequality likewise generalizes.  This theorem 
will not be needed in the rest of  this paper, and we do not discuss this here.
\end{remark}

\begin{remark} The fact that this method of proof cannot yield the inequality for all $s,t$, even in cases such as those described below for which the inequality is true for all $s,t$, has nothing to do with what is given up in the application of the Lieb-Thirring inequality: Consider the case $s=t$, $Z=X$ and $W= Y$. Then \eqref{holdlim} becomes
\begin{equation}\label{holdlim2}
\left|  \tr[X^tY^t X^{1-t}Y^{1-t}]    \right|  \leq  \|X^{2t-1}Y^{2t-1}\|_{1/(2t-1)}  \|Y^{1-t}X^{1-t}\|_{1/(1-t)}  \|Y^{1-t}X^{1-t}\|_{1/(1-t)}\  .
\end{equation}
Hence for $X,Y>0$,
$$
\lim_{t \uparrow1} \|X^{2t-1}Y^{2t-1}\|_{1/(2t-1)}  \|Y^{1-t}X^{1-t}\|_{1/(1-t)}  \|Y^{1-t}X^{1-t}\|_{1/(1-t)} = \|XY\|_1\ ,
$$
and in general, $\|XY\|_1 > \tr[XY]$.
\end{remark}

We now present several results that provide conditions on $X$ and $Y$  under which \eqref{1.1} and \eqref{1.1b}  are valid for all $s,t\in [1/2,1]\times [1/2,1]$.  We will use the following lemma:

\begin{lm}\label{conlem} Suppose that $X$ and $s$ are such that in a basis in which $Y$ is diagonal,
\begin{equation}\label{convex}
(X^s)_{i,j}(X^{1-s})_{j,i} \geq 0 \quad{\rm for\ all}\qquad i,j\ .
\end{equation}
Then for all $t\in [1/2,1]$, 
\begin{equation}\label{convex2}
 \tr] X^{1/2}Y^{1/2}X^{1/2} Y^{1/2}] \leq |\tr  [X^sY^t X^{1-s}Y^{1-t}] |   \leq \tr [ XY]\ .
\end{equation}
\end{lm}

\begin{remark}\label{Hadamard} The matrix $M_{i,j} := (X^s)_{i,j}(X^{1-s})_{j,i}$ is the Hadamard product of two positive matrices, namely $X^s$ and the transpose of $X^{1-s}$, and as such it is positive semidefinite. However the off-diagonal entries need not be positive or even real. 
\end{remark}

\begin{proof}   Assume first that $Y> 0$. Computing in any basis that diagonalizes $Y$, with the $j$th diagonal entry of $Y$ denoted by $y_j$,
$$
f(t) := \tr [X^sY^t X^{1-s}Y^{1-t} ] = \sum_{i,j} \left((X^s)_{i,j}(X^{1-s})_{j,i} \right) y_j^t y_i^{1-t}\ ,
$$
where now it is convenient to let $t$ range over $[0,1]$. Under the hypothesis \eqref{convex},  $f(t)$ is symmetric and convex in $t$. Hence its maximum occurs at $t=0$ and $t=1$, and its minimum occurs at $t=1/2$  Since $Y> 0$, $\lim_{t\uparrow  1}\tr  X^sY^t X^{1-s}Y^{1-t}  = \tr X^sY X^{1-s}  = \tr\ XY$. 
This proves that $\tr  X^sY^t X^{1-s}Y^{1-t} $ is real and satisfies 
$$
\tr [X^{s}Y^{1/2}X^{1-s} Y^{1/2}] \leq |\tr  [X^sY^t X^{1-s}Y^{1-t}] |   \leq \tr[ XY]
$$
Since $(Y^{1/2})_{i,j}(Y^{1/2})_{j,i} = |Y^{1/2}_{i,j}|^2$, we may now apply what was proved above with the roles of  $X$ and $Y$ interchanged to conclude that
$$
 \tr [X^{1/2}Y^{1/2}X^{1/2} Y^{1/2}] \leq     \tr [X^{s}Y^{1/2}X^{1-s} Y^{1/2}]\ .
$$
Finally, we obtain the same result assuming only $Y\geq 0$ using the obvious limiting argument. 
\end{proof}

Our first application of Lemma~\ref{conlem} is to pairs of operators of a sort that arise frequently in mathematical physics. For $X > 0$, define $H = -\log(X)$ so that  $X = e^{-H}$.  Suppose that in a basis in which $Y$ is diagonal, all off-diagonal entries of $H$ are non-positive; i.e.,
\begin{equation}\label{BD1}
H_{i,j} \leq 0 \quad{\rm for\ all}\quad i\neq j\ . 
\end{equation}
For example, this is the case if $H$ is the graph Laplacian on an unoriented graph (with the graph theorist's sign convention that the graph Laplacian is non-negative); see Example~\ref{GL} below.

It is well-known that under these conditions, as a consequence of the Beurling-Deny Theorem, \cite[Theorem 5]{BD}, the semigroup $e^{-sH}$ is positivity preserving, and so in particular $(e^{-sH})_{i,j} \geq 0$ for all $s$ and all $i,j$.
For the readers convenience, we recall the relevant part of their proof adapted to our setting:  Take $\lambda > 0$ sufficiently small that $I+ \lambda H$ is invertible.  Then for any vector $f$,
$\left(1 + \lambda H\right)^{-1}f$
is the unique minimizer of
$$
F(u) := \lambda \langle u, H u\rangle + \|u - f\|^2 \ ;
$$
the uniqueness follows from the strict convexity of $F$ for sufficiently small $\lambda > 0$. 
Under the condition \eqref{BD1},  when $f = |f|$, $F(|u|) \leq F(u)$.  Hence  $\left(1 + \lambda H\right)^{-1}f$ maps the positive cone into itself, and all entries of this matrix are non-negative. The same is evidently true of $\left(1 + \lambda H\right)^{-n}f$ for all $n$. Taking $\lambda = s/n$ and $n \to \infty$, the same is true of $e^{-sH}$ for all $s \geq 0$. 

\begin{thm}\label{BDthm}  Suppose that $H = -\log X$ satisfies \eqref{BD1} in a basis in which $Y$ is diagonal. Then \eqref{1.1} and \eqref{1.1b} are valid for all $s,t\in [1/2,1]\times[1/2,1]$. 
\end{thm}

\begin{proof}  By the Beurling-Deny Theorem as explained above, for all $s> 0$
$$
(X^s)_{i,j} = (e^{-sH})_{i,j} \geq 0\ .
$$
It follows that \eqref{convex} is satisfied for all $s$, and now the conclusion follows from Lemma~\ref{conlem}.
\end{proof}

One may also use Lemma~\ref{conlem} to show that both \eqref{1.1} and \eqref{1.1b} are valid for $2\times 2$ matrices, as was already shown in \cite{AHO}:
Write 
$X = \left[\begin{array}{cc} a & z\\ \overline{z} &b\end{array}\right]$.  Then by the usual integral representation formula for $X^s$, $0 < s < 1$, 
$$(X^s)_{1,2} =  -z \left(\frac{\sin(\pi\alpha)}{\pi}\int_0^\infty \lambda^{s}  \frac{1}{(a+\lambda)(b+\lambda) - |z|^2}   {\rm d}\lambda\right)
$$
showing that for all $0 < \alpha < 1$, $(Y^\alpha)_{1,2}$ is a positive multiple of $-z$, and hence \eqref{convex} is always true. 

Our next theorem provides another class of examples of positive matrices $X$ and $Y$ for which \eqref{1.1} is true for all $1/2 \leq s,t \leq 1$.  A related theorem, for a version of \eqref{1.1} with  the operator norm in place of the trace, has recently been proved in \cite{ACPS} by quite different means.

\begin{thm}  Let $H$ and $K$ be arbitrary self-adjoint $n\times n$ matrices. Then there exists an $\alpha_0 > 0$ depending on $H$ and $K$ so that for all $\alpha< \alpha_0$,
with $X := e^{\alpha H}$ and $Y := e^{\alpha Y}$, \eqref{1.1} is valid  all $1/2 \leq s,t \leq 1$. 
\end{thm}

\begin{proof}  If $H$ and $K$ commute, then it is obvious that  \eqref{1.1} is valid  all $1/2 \leq s,t \leq 1$, no matter what $\alpha>0$ may be. Hence we may assume without loss of generality that $[H,K] \neq 0$.   Also without loss of generality, we may  suppose that $H$ and $K$ are both contractions and $0 \leq \alpha \leq 1$.
Then by the spectral theorem,
$$
\left\Vert e^{\alpha H} - \left( I + \alpha H + \frac{\alpha^2}{2} H^2\right)\right\Vert  \leq e^{\alpha}\frac{\alpha^3}{6}\ ,
$$
and likewise for $K$, 
Thus
\begin{equation}\label{BCH1}
\left\Vert e^{\alpha H}e^{\alpha K}  - \left( I + \alpha H + \frac{\alpha^2}{2} H^2\right)  \left( I + \alpha K + \frac{\alpha^2}{2} K^2\right)\right\Vert  \leq e^{2\alpha}\frac{\alpha^3}{2}\ .
\end{equation}
Note that 
\begin{multline}\label{BCH2}
\left( I + \alpha H + \frac{\alpha^2}{2} H^2\right)  \left( I + \alpha K + \frac{\alpha^2}{2} K^2\right) =\\
I + \alpha(H+K) + \frac{\alpha^2}{2} (H+K)^2 +  \frac{\alpha^2}{2} [H,K] + R\ ,
\end{multline}
where $\|R\| \leq 3 \alpha^3$. 

Now writing $X = e^{\alpha H}$ and $Y = e^{\alpha K}$, 
$$
\tr[X^{1-s}Y^{1-t}X^s Y^t] =  \tr[XYZ]\ \qquad{\rm where}\qquad  Z(s,t) := Y^{-t}X^sY^tX^{-s}\ .
$$
Using \eqref{BCH1} and \eqref{BCH2},
we obtain 
\begin{equation}\label{stid}
\|Z(s,t)  -( I + \alpha^2 st[H,K])\| \leq C\alpha^3\ ,
\end{equation}
for some constant $C$ that can be easily estimated. 
Note that for all $s,t$, $Z(0,t) = Z(s,0) = I$.  For this reason, there cannot have been any terms proportional to $s^2$ or $t^2$ in the second order expansion.

Altogether we have 

\begin{eqnarray*}
\tr[X^{1-t}Y^{1-s}X^t Y^s]  &=& \tr\left[ \left(I + (K+H) + \frac12(K+H)^2 + \frac12 [H,K]\right) \left(I + st[H,k] \right)\right] +R_2\\
&=&  \tr[XY]\\
&+& st \tr\left[   \left(I + \alpha(K+H) + \frac{\alpha^2}{2}(K+H)^2 + \frac{\alpha^2}{2} [H,K]\right)[H,K] \right]  +R_3\ .
\end{eqnarray*}
where $\|R_2\|,\|R_3\| \leq C\alpha^3$ for some constant $C$. 
Evidently, $\tr[[H,K]] = \tr[[H[H,K]] = \tr[K[H,K]] =  \tr[[H^2[H,K]] = \tr[K^2[H,K]]  = 0$.
A simple computation shows that
$$
\tr[(HK+KH)(HK-KH)] = 0\ .$$
Therefore,
$$
\tr[X^{1-s}Y^{1-t}X^s Y^t]  -  \tr[XY] + st \alpha^2 \tr[H,K]^2 + \tr[R_4]\ .
$$
where $\|R_4\| \leq C\alpha^3$, and hence $\tr[R_4] \leq nC\alpha^3$. 
Evidently, since by hypothesis $[H,K] \neq 0$, 
$\tr[H,K]^2  < 0$.  Thus for all $\alpha$ sufficiently small, $\tr[X^{1-t}Y^{1-s}X^t Y^s]  -  \tr[XY] < 0$ for all $(s,t)\in [1/2,1]\times[1/2,1]$. 
\end{proof}

Of course, replacing $t$ by $1-t$ and $s$ by $1-s$, the same proof shows, with the same $\alpha_0$ that when $\alpha \leq \alpha_0$,
$$
\tr[X^{1-t}Y^{1-s}X^t Y^s]  =  \tr[XY] + st \alpha^2 \tr[H,K]^2  \pm C \alpha^3\ .
$$
Replacing $s$ by $is$ and $t$ by $it$ yields 
$$
\tr[X^{1-is}Y^{1-it}X^{is} Y^{it}]   = \tr[XY] - (st)^2 \tr[H,K]^2 + \mathcal{O}(\delta^6)\ .
$$
and hence
$[X,Y] \neq 0$,  and $\alpha$ sufficiently small, 
$$
\tr[X^{1-it}Y^{1-is}X^{it} Y^{is}]   > \tr[XY]\ .
$$
Thus the three lines argument in \cite{BELV} cannot hold for $s,t$ sufficiently close to $1$ or $0$.

\section{The monotonicity of the Golden--Thompson inequality}

Let $H$ and $K$ be self-adjoint $n\times n$ matrices. For $0 \leq u \leq 1$, define
\begin{equation}\label{fofu}
f_{H,K}(u) = \tr[e^{H+ (1-u)K}e^{uK}]\ .
\end{equation}
Then  $f(0) = \tr[e^{H+K}]$ and $f(1) = \tr[e^He^K]$, and by the Golden-Thompson inequality,
\begin{equation}\label{gt1}
\tr[e^{H+K}] \leq\tr[e^He^K]\ ,
\end{equation}
$f_{H,K}(0) \leq f_{H,K}(1)$. In this section we ask: When is $f_{H,K}(u)$ monotone increasing in $u$?  We shall prove that this is the case for an interesting class of pairs $(H,K)$ of self-adjoint matrices, and we shall show that it is not true in general.

\begin{remark}\label{sign}Observe that if one replaces $H$ by $H+aI$ and $K$ by $K+bI$,
\begin{equation}\label{gt2}
f_{H+aI,K+bI}(u) = e^{a+b}f_{H,K}(u)\ ,
\end{equation}
and hence whether or not  $f_{H+aI,K+bI}(u)$ is monotone increasing is independent of $a$ and $b$.  
\end{remark}

\begin{thm}\label{GTThm} Suppose that $K$ is diagonal and that  all off-diagonal entries of $H$  are non-negative.  Then $f_{H,K}(u)$ is monotone increasing. 
\end{thm}

\begin{proof}  By Remark~\ref{sign}, we may assume that  $K \geq 0$.  It will be convenient to define $H_u = H + (1-u)K$.  Then 
\begin{eqnarray}\label{fp}
f'(u) &=&\tr[e^{H_u} K e^{uK}] - \int_0^1\tr[e^{(1-t)H_u} K e^{t H_u} e^{uK}]{\rm d}t \nonumber\\
&=& \sum_{m=0}^\infty\frac{u^m}{m!}  \left( \tr[e^{H_u} K^{m+1}] - \int_0^1 \tr[e^{(1-t)H_u} K e^{t H_u} K^m]{\rm d}t  \right)\ .
\end{eqnarray}
Now define $X = e^{H_u}$ and for each $m$, $Y = K^{m+1}$ and $s = (m+1)^{-1}$. With these definitions,
\begin{eqnarray*}
\tr[e^{H_u} K^{m+1}] - \int_0^1 \tr[e^{(1-t)H_u} K e^{t H_h} K^m]{\rm d}t &=& \tr[BA] - \int_0^1\tr [B^{1-t} A^s B^{t} A^{1-s}]{\rm d}t \ .\\
\end{eqnarray*}

Since $Y$ is diagonal, for each $u$, $-\log H_u$ has non positive off diagonal entries. By Theorem~\ref{BDthm},
$$
\tr[BA] - \int_0^1\tr [B^{1-t} A^s B^{t} A^{1-s}]{\rm d}t  \geq 0\ .
$$
Then by \eqref{fp}, $f'(u) \geq 0$. 
\end{proof}

\begin{exam}\label{GL} Let $\mathcal{G}$ be a graph with a finite set  of vertices $\mathcal{V}$. Let the edge set be $\mathcal{E}$; this is a subset of 
$\mathcal{V}\times \mathcal{V}$.  Suppose that $\mathcal{G}$ is a simple graph, meaning that $(x,x) \notin  \mathcal{E}$ for all $x\in \mathcal{V}$, and that 
$(x,y) \in \mathcal{E}$ if and only if $(y,x) \in \mathcal{E}$. Then the graph Laplacian, $\Delta_{\mathcal{G}}$ is defined by
$$
\Delta_{\mathcal{G}}f(f) = \sum_{\{ y\ :\ (x,y)\in \mathcal{E}\}}(f(x) - f(y))\ .
$$
In the natural basis, all off diagonal elements of the matrix representing $\Delta_{\mathcal{G}}$ are non-positive. Define $H_0 = \Delta_{\mathcal{G}}$, to obtain a non-negative ``free Hamiltonian'' as in the usual mathematical physics convention. 
Let $V$ be a self adjoint multiplication operator on $L^2(\mathcal{V}, \mu)$, where $\mu$ is the uniform probability measure on $\mathcal{V}$. In the natural basis, $V$ is diagonal.

Then by Theorem~\ref{GTThm},  
$$
f(u) :=  \tr[ e^{-(H_0+ (1-u) V)}e^{-uV}]
$$
is strictly monotone increasing in $u$. 
\end{exam}

\begin{exam}  Though we have given proofs in the context of matrices, it is is easy to see that the proofs extend to cover interesting infinite dimensional cases. Let $X = e^{\beta \Delta}$ where $\Delta$ is the Laplacian on $\R^d$ and $\beta > 0$. Let $V$ be a real valued function on $\R^d$, and let $V$ also denote multiplication by $V$ acting on $L^2(\R^d)$, which is in general unbounded. Let $Y = e^{-\beta V}$.  Then since $X^t$ has a positive kernel and $Y$ acts by multiplication on $L^2(\R^d)$, the proof of Theorem~\ref{GTThm}  is easily adapted to show that
$$
f(u)  :=   \tr[e^{-\beta (\Delta+(1-u)V)}e^{-\beta u V}]
$$
is monotone increasing in $u$. The same applies with $-\Delta$ replaced by $(-\Delta)^{1/2}$, another case that arises in physical applications. 
\end{exam}

\section{Counterexamples}

This section presents  the constructions of counter-examples showing that the inequalities \eqref{1.1} and \eqref{1.1b} cannot hold in general, even in the $3\times 3$ case, and showing  the monotonicity property established in Theorem~\ref{GTThm} under specified conditions  cannot hold in general. While counterexamples for 
 \eqref{1.1} and \eqref{1.1b} were found by Plevnik \cite{P}, our goal is to provide a systematic approach to their construction.   
 Plevnik provided two completely separate and purely numerical  counter-examples  to \eqref{1.1} and \eqref{1.1b}. 
 We provide a method for constructing a family of counter-examples that goes further in significant ways. For example, while Plevnik showed in 
 \cite[Example 2.5]{P} that \eqref{1.1b} can be violated, his example does not show that it is possible for $\tr[X^sY^yX^{1-s}Y^{1-t}]$ to be negative.  
 We show that this is the case.  Moreover, our construction shows that the failure of  the inequalities \eqref{1.1} and \eqref{1.1b}  as well as the 
 failure in general of the  monotonicity of the Golden-hompson Inequality described  in  Theorem~\ref{GTThm} are all closely connected: 
 Essentially one example undoes all three would-be conjectures.
 
We have seen in Lemma~\ref{conlem} that that if all of the entries of  $M_{i,j} := (X^s)_{i,j} (X^{1-s})_{j,i}$ are nonnegative, then   \eqref{1.1} and 
\eqref{1.1b} both hold. In constructing our counter-examples, we shall take $X$ to be real, and hence the entries of $M$ will be real for each $s$
 
 \begin{lm}\label{sym}
Let $y := (y_1,\dots y_n)$ be any vector in $\R^n$, Let $X$ be any positive semidefinite $n\times n$ matrix matrix, and let $0 \leq t \leq 1$. Let $M(s)$ 
denote the matrix $M_{i,j}(s) := (X^s)_{i,j} (X^{1-s})_{j,i}$.  Then for all $0 < s < 1$, 
\begin{equation}\label{oct}
\sum_{i,j=1}^n   M_{i,j}(s) (y_i - y_j)^2 \geq 0\  .
\end{equation}
\end{lm}

\begin{proof} We may assume that the entries of  $y$ are positive since the left side of \eqref{oct} does not change when we add to $y$ any multiple of the vector each of whose entries is $1$. 

By Lemma~\ref{conlem} we know that for $X$ and any matrix $Y\geq 0$ (we replace $Y$ by $Y^2$ in Lemma~\ref{conlem} for convenience),
$$\tr [ X^{1-s} Y X^{s} Y]   \leq \tr[XY^2] = \tr[Y^2X]\ .$$  
Letting $Y$ be the diagonal matrix whose $j$th diagonal entry is $y_j$, this becomes
$$
\tr [ X^{1-s} Y X^{s} Y] = \sum_{i,j=1}^n  y_iM_{i,j}(s) y_j   =  \sum_{i,j=1}^n  y_jM_{i,j}(s) y_i \leq \sum_{i,j=1}^n  y_i^2M_{i,j}(s)  = 
\sum_{i,j=1}^n  M_{i,j}(s) y_j^2 \ .    $$
\end{proof}

We now claim that if $X\geq 0$ is a real $3\times 3$ matrix, for any $0<s< 1$, $M(s)$ has at most one entry above the diagonal that is negative. 
(By Remark \ref{Hadamard}, all diagonal entries are non-negative, and $M(s)$ is symmetric, so the same is true below the diagonal.)   
 To see this, take the vector $y$ to be of the form $(0,1,1)$, $(1,0,1)$ of $(1,1,0)$. Then for these choices,  \eqref{oct} becomes
\begin{equation}\label{oct2}
2(M_{1,2}(s) + M_{1,3}(s))\geq 0\  ,\quad   2(M_{1,2}(s) + M_{2,3}(s)) \geq 0 \ \quad{\rm and}\quad 2(M_{1,3}(s) + M_{2,3}(s)) \geq 0\ .
\end{equation}
Thus each pair of entries above the diagonal must have a non-negative sum, and hence no two can be negative. 

One might hope that one could construct counter-examples to \eqref{1.1} and \eqref{1.1b} by constructing matrices $X>0$ for which $M_{i,j}(t) < 0$ for all
$t\in (0,1/2)\cup(1/2,1)$. This is easy to do, but this alone does not yield counterexamples. 

For, example, define
${\displaystyle
X^{1/2} = \left[\begin{array}{ccc} 2 & \sqrt{2} & 0\\ \sqrt{2} & 2 & \sqrt{2}\\ 0 & \sqrt{2} & 2\end{array}\right]}$.
This matrix is easily diagonalized; the eigenvalues are $4$, $2$ and $0$. Since $X^{1/2}_{1,3} =0$, one might expect that $X^s_{1,3}$ changes sign at $s=1/2$, and only there, so that $M_{1,3}(s) \leq 0$ for all $0 < s < 1$. Indeed, doing the computations, one finds
\begin{equation}\label{ce1}
M_{1,3}(s) = -\frac{4^{-s}}{2}\left( 4^s -2 \right)^2 \leq 0 \quad{\rm while}\quad  M_{1,1}(s) = M_{3,3}(s) = \frac{4^{-s}}{2}\left( 4^s +2 \right)^2 
\end{equation}
 Now take $Y := \left[\begin{array}{ccc} a &0&0\\ 0 & 0 & 0\\ 0 & 0 & b\end{array}\right]$ with $a,b > 0$ and distinct. 
Then
\begin{equation}\label{ce2}
\tr [ X^{1-s} Y^{1-t} X^{s} Y^t]  =  M_{1,1}(s) a + M_{3,3}(s)b  +  M_{1,3}(s)(a^{1-t}b^t + a^t b^{1-t})\ .
\end{equation}
For fixed $s\notin\{0, 1/2,1\}$, this is strictly concave in $t$ and symmetric about $t=1/2$, 
so the maximum occurs only  at $t=1/2$, and the minimum only at $t\in \{0,1\}$.
However, since $\lim_{t\downarrow 0}Y^t = P :=  \left[\begin{array}{ccc} 1 &0&0\\ 0 & 0 & 0\\ 0 & 0 & 1\end{array}\right] \neq I$, 
we do not have 
$\lim_{t\downarrow 0} \tr [ X^{1-s} Y^{1-t} X^{s} Y^t] = \tr[XY]$,
which would provide a counterexample to \eqref{1.1}, but instead
$
\lim_{t\downarrow 0} \tr [ X^{1-s} Y^{1-t} X^{s} Y^t] = \tr [ X^{1-s} Y X^{s} P]$.
As we have just seen, this is less than $ \tr [ X^{1-s} Y^{1/2} X^{s} Y^{1/2}]$, and by Lemma~\ref{conlem}, this in turn is less than $\tr[XY]$.  In fact,
defining $h(t) := 4^{t-1/2} + 4^{1/2- t}$,  we can rewrite \eqref{ce1} as
\begin{equation}\label{ce14}
M_{1,3}(s) =2 - h(s) \quad{\rm and}\quad  M_{1,1}(s) = M_{3,3}(s) = 2 + h(s) \ .
\end{equation}
Then from \eqref{ce2},
\begin{equation}\label{ce5}
\tr [ X^{1-s} Y^{1-t} X^{s} Y^t]  =  2(a+b+ a^{t}b^{1-t} + a^{1-t}b^{t}) + h(s)(a+b - a^{t}b^{1-t} - a^{1-t}b^{t}) \ .
\end{equation}
By the arithmetic-geometric mean inequality, $a+b - a^{t}b^{1-t} - a^{1-t}b^{t} \geq 0$ for all $0 \leq t \leq 1$.  Since $h(s)$ is evidently convex 
and symmetric about $s=1/2$, for each fixed $t \in (0,1)$, $\tr [ X^{1-s} Y^{1-t} X^{s} Y^t]$ is a strictly convex function of $s$, symmetric 
about $s=1/2$. Therefore this function is  minimized only for $s= 1/2$  and maximized only for $s\in \{0,1\}$  and hence  for any $t$,
$$\tr [ X^{1-s} Y^{1-t} X^{s} Y^t]  \geq  \tr [ X^{1/2} Y^{1-t} X^{1/2} Y^t], $$
and the right side is independent of $t$ since $X^{1/2}_{1,3}  =0$. Hence  \eqref{1.1b} is satisfied for all choices of $a,b > 0$. Likewise,
by what was proved above, for all $s,t$,  with $Q := \lim_{s\downarrow 0}X^s$, which is an orthogonal projection,
$$\tr [ X^{1-s} Y^{1-t} X^{s} Y^t]  \leq  \tr [ X^{1-s} Y^{1/2} X^{s} Y^{1/2}] \leq  \tr [ Q Y^{1/2} X^{1} Y^{1/2}] \leq \tr[XY]\ ,$$
and hence 
 \eqref{1.1} is satisfied for all choices of $a,b > 0$. 
 
 This shows that the construction of counterexamples is more subtle than simply producing negative entries in $M(s)$.   It appears that the 
 key to the construction of counterexamples for $3\times3$ matrices is to choose $X$ so that one of the inequalities in \eqref{oct} to is nearly 
 saturated, with one of the summands negative for most values of $s$. Furthermore, it is natural to choose $X$ and $Y$ to be perturbations of 
 positive semidefinite matrices $X_0$ and $Y_0$ such that $\tr [ X_0^{1-s} Y_0^{1-t} X_0^{s} Y_0^t] = 0$ for all $0 \leq s,t \leq 1$.  
 Of course this is satisfied if $X_0$ and $Y_0$ are orthogonal projections with mutually orthogonal ranges. 
 
 Our construction relies on the {\em Householder reflections} determined by two distinct unit vectors $u,v\in \R^n$. This is given by
 $H_{u,v} := I - 2\|u-v\|^{-1}| u-v\rangle \langle u-v|$.
 Evidently, $H_{u,v}$ is self adjoint, orthogonal, and  $H_{u,v} u =v$ and $H_{u,v}v = u$. {\em For simplicity}, choose 
 \begin{equation}\label{ce21}
 u := (0,0,1)\qquad{\rm  and }\qquad
 v := 2^{-1/2}(1,1,0)\ .
 \end{equation}
  Then
 $$
 U :=  H_{u,v} = \frac12 \left[\begin{array}{ccc} \phantom{-}1 & -1& \sqrt{2} \\ -1 & \phantom{-}1& \sqrt{2} \\ \sqrt{2} & \sqrt{2}& 0\end{array}\right]
 $$
 Now choose
 $$
 Y_0 := \left[\begin{array}{ccc} 0 & 0 & 0\\ 0 & 0 & 0 \\ 0 & 0 & 1\end{array}\right]\quad{\rm and}\quad X_0 = UY_0U = \frac12
 \left[\begin{array}{ccc} 1 & 1 & 0\\ 1 & 1 & 0 \\ 0 & 0 & 2\end{array}\right]\ .
 $$
 Then $X_0$ and $Y_0$ are orthogonal projections such that $X_0Y_0 = 0$. 
 
 Now we make a simple perturbation. For $a,b > 0$, small, to be chosen later,  define
 $$
 A := \left[\begin{array}{ccc} a & 0 & 0\\ 0 & b & 0 \\ 0 & 0 & 1\end{array}\right]\quad{\rm and}\quad  Y :=  \left[\begin{array}{ccc} c & 0 & 0\\ 0 & d & 0 \\ 0 & 0 & 1\end{array}\right]
 $$
 and also for $0 < t < 1$, define
 $$
 \alpha := \frac14(a^t+b^t) \qquad{\rm and}\qquad \beta :=  \frac{\sqrt{2}}{4}(a^t-b^t) \ .
 $$
 Then
 $$
 UAU = X_0 + \left[\begin{array}{ccc}  \phantom{-}\alpha & - \alpha & \phantom{-}\beta\\ -\alpha  & \phantom{-}\alpha &-\beta\\
\phantom{-} \beta & -\beta & \phantom{-}2\alpha\end{array}\right]\ .
 $$
 The off-diagonal entries of $UYU$ will not change sign as $t$ varies, but we can make this happen by applying are further orthogonal transformation; define
 $R := \left[\begin{array}{ccc} \cos x & 0 &- \sin x\\ 0 & 1 & 0\\ \sin x & 0 & \phantom{-}\cos x\end{array}\right]$, and finally put
 $$X := RUAUR^T\ ,$$
 where $R^T$ is the transpose of $R$, with $x$, $a$ and $b$ to be chosen later. We compute
 $$
 X^t_{1,3}  = (\cos^2 x - \sin^2 x)\beta(t)  + \sin x \cos x \left(\frac12 - \alpha(t)\right)\ .
 $$
 and
 $$
 X^t_{2,3}  = -\cos(x)\beta(t)  + \sin x \left(\frac12 - \alpha(t)\right)\ .
 $$
 
 We seek a small perturbation of $X_0$, and hence we will take $a$, $b$ and $|x|$ all to be small positive numbers.   It is easy to see that the sign change we seek occurs in $X^t_{1,3}$  if we take $a \ll b \ll 1$ and $0 <x \ll 1$, and occurs in  $X^t_{2,3}$  if we take $b \ll a \ll 1$ and $0 <x \ll 1$

\begin{exam}\label{11}
To get a counterexample to \eqref{1.1}, take $a=10^{-10}$, $b= 10^{-19}$, $x=10^{-5}$, $c=10^{-10}$ and $d= 0$. Then one finds
$$
\tr[XY] < 1.50001 \times 10^{-10}  \quad{\rm while} \quad \tr[X^{0.79}Y^{0.79}X^{0.21}Y^{0.21} > 1.61022 \times 10^{-10} \ .
$$
\end{exam}
 
 \begin{exam}\label{11b}
 To get a counterexample to \eqref{1.1b}, take $a=10^{-19}$, $b= 10^{-10}$, $x=10^{-5}$, $c=10^{-10}$ and $d= 0$. Then one finds
 $$
 \tr[X^{0.98}Y^{0.98}X^{0.02}Y^{0.02}] < -2.38674\ ,
 $$
 which being negative, is certainly less that $\tr[X^{1/2}Y^{1/2}X^{1/2}Y^{1/2}] > 0$, and by continuity, somewhere the trace must be  zero.
 \end{exam}
 
  Notice that the  only difference between the two examples is 
 that we have swapped the values assigned to $a$ and $b$; all other parameters are left the same. Numerical plots show that in both cases, 
 the maximum value of
 $|X^t_{1,3} + X^t_{2,3}|$ is less that $10^{-3}$ times the maximum of  $|X^t_{1,3}| + |X^t_{2,3}|$,  so that the last inequality in \eqref{oct} is 
 nearly saturated; there is near cancellation in the sum $X^t_{1,3} + X^t_{2,3}$.   Notice that in our counterexample to \eqref{1.1}, the sum of the 
 exponents $s+t$ is $1.58$, not so much larger than the minimum value, $3/2$, at which such a counterexample cannot exist. It would be of interest to 
 see if one can build on this construction, possibly extending it into higher dimensions,  to show that the condition $s+t \leq 3/2$ in 
 Theorem~\ref{AHO1} is sharp. 
 
 We close by showing that the monotincity property for the Golden-Thompson Inequality descried in Theorem~\ref{GTThm} does not hold for arbitrary self adjoint matrices $H$ and $K$. 
 
 Recall that $f_{H,K}(u) $ has been defined by \eqref{fofu}
 \begin{equation}\label{fofu2}
f_{H,K}(u) = \tr[e^{H+ (1-u)K}e^{uK}]\ .
\end{equation}
$$
\frac{{\rm d}}{{\rm d}u}f_{H,K}(u)\bigg|_{u=1} = \tr[e^{H}Ke^{K}] - \int_0^1 \tr[e^{tH} K e^{(1-t)H} e^K]{\rm d}t\ .
$$
With $X$ and $Y$ as above, we define $K = \log(X)$, and $H = \log Y$. Since $H$ is diagonal, the integral $ \int_0^1 e^{tH} K e^{(1-t)H}{\rm d}t $ can be explicitly evaluated as a Hadamard product. One finds
$$
\frac{{\rm d}}{{\rm d}u}f_{H,K}(u)\bigg|_{u=1} < -3\times 10^{-6}\ ,
$$ 
This shows that the monotonicity proved in Theorem~\ref{GTThm} is not true for general self-adjoint $H$ and $K$.

\bigskip
\noindent{\bf \Large{Acknowledgements}}  

We thank Victoria Chayes and Rupert Frank for useful conversations.


\begin{thebibliography}{99}


\footnotesize{

\medskip

\bibitem{ACPS}   R.~Alaifari, X.~Cheng, L.~B.~Pierce and S.~Steinerberger, \textit{On matrix rearrangement inequalities}, Proc. Amer. Math. Soc. {\bf 148} (2020), 1835--1848

\bibitem{AHO} T.~Ando, F.~Hiai and K.~Okubo: \textit{Trace inequalities for multiple products of two matrices}, Math. Inequ. and Appl, {\bf 3}, No. 3,  307--318(2000)

\bibitem{A90} H.~Araki, \textit{On an inequality of Lieb and Thirring} Letters in Math Phys. {\bf 19} (1990),167--170.

\bibitem{BD} A.~Beurling and  J.~Deny  \textit{Espaces de dirichlet: I. Le cas \'el\'mentaire}  Acta Math. {\bf 99} (1958), 203--224 

\bibitem{BELV} T.~Bottazzi, R.~Elencwajg, G.~Larotonda and A.~Varela, \textit{Inequalities related to Bourin and Heinz means
with a complex parameter}, J. Math. Anal. Appl., {\bf 426} (2015), 765--773.


\bibitem{FS} S.~Frieedland and W.~So, \textit{On the product of matrix exponentials}, Lin. alg. Appl. {\bf 196} (1994), 193--205


\bibitem{G65} S.~Golden, \textit{Lower bounds for the Helmholtz function}, Phys. Rev., Series II, {\bf 137} (1965) B1127--B1128


\bibitem{LT76} E.~H.~Lieb, W.~E.~Thirring, \textit{Inequalities for the Moments of the Eigenvalues of the Schr\"odinger Hamiltonian and Their Relation to Sobolev Inequalities}, pp. 269--303 in \textit{Studies in Mathematical Physics}, eds. E.~Lieb, B.~Simon, and A.~Wightman, Princeton University Press,  Princeton, 1976.

\bibitem{P} L.~Plevnik: \textit{On a matrix trace inequality due to Ando,
Hiai and Okubo}, Indian J. Pure and Appl. Math., {\bf 47}, 491–500 (2016).
DOI: 10.1007/s13226-016-0180-947.


\bibitem{S05} B.~Simon, \textit{Trace Ideals and Their Applications: Second Edition}, Mathematical Surveys and Monographs, {\bf 120}, AMS, Providence RI 2005


\bibitem{T65}  C.~J.~Thompson, \textit{Inequality with applications in statistical mechanics}, Jour. of Math. Physics, {\bf 6}  (1965) 181--1813



}


\end{thebibliography}
\end{document}